\documentclass{ifacconf}

\usepackage{graphicx}      
\usepackage{natbib}        

\usepackage{graphicx}

\usepackage{savesym}
\savesymbol{AND}
\usepackage{algorithm, algorithmic, setspace}
\usepackage{graphicx} 
\usepackage{amsmath} 
\usepackage{amssymb}  
\usepackage{amsfonts}
\usepackage{color}
\usepackage[english]{babel}

\newcommand{\argmin}[1]{\underset{#1}{\mathrm{argmin\,}}}

\newcommand{\atrue}{\widetilde{a}}

\newcommand{\Aa}{\widetilde{A}}
\newcommand{\Ca}{\widetilde{C}}

\newcommand{\sign}{\text{sign}}

\newcommand{\zstar}{z^{\star}}
\newcommand{\reg}{r}

\newcommand{\R}{\mathbb{R}}

\newcommand{\Z}{\mathbb{Z}}
\newcommand{\fun}{\mathcal{F}}

\newcommand{\soft}{\mathrm{S}}

\renewcommand{\O}{\mathcal{O}}

\newcommand{\irr}{\rho}

\newcommand{\xhat}{\hat{x}}
\newcommand{\ahat}{\hat{a}}
\newcommand{\zhat}{\hat{z}}

\newtheorem{proposition}{Proposition}
\newtheorem{problem}{Problem}

\newtheorem{assumption}{Assumption}

\begin{document}
\begin{frontmatter}

\title{Online sparse observers for cyber-physical systems under sensor bias}
%
%
\author[Poli]{V. Cerone}
\author[Poli]{S. M. Fosson}
\author[Poli]{D. Regruto}
\author[Poli]{F. Ripa}
\address[Poli]{Department of Control and Computer Engineering,\\ Politecnico di Torino, Italy (e-mail: sophie.fosson@polito.it).}

\begin{abstract}.
The design of state observers for cyber-physical
systems under sparse sensor biases, faults, or attacks has drawn substantial attention in recent years.
While sparsity-based batch approaches, which collect multiple measurements and run offline,   are relatively
mature, online secure state estimation, for real-time state/attack recovery, is still an open problem. Although some algorithms have been proposed, convergence guarantees remain limited, even for the case of constant attacks.

As a first step toward addressing this gap, we focus on constant attacks. We analyze the observability of this setting and we study several online sparse observers, derived from different methodological frameworks such as online sparse optimization and block Bregman methods. Some of the proposed observers are adaptations of existing algorithms to the secure state-estimation setting, while others constitute novel algorithmic contributions.

The goal of this paper is to provide a unified overview of practically implementable, online sparse observers, discuss their convergence properties, and compare their performance through numerical experiments.

%
%
\end{abstract}

\end{frontmatter}

\section{Introduction}\label{sec:IN}

Cyber-physical systems (CPSs) are complex dynamical systems with a distributed architecture in which physical processes are integrated with cyber layers. Within these setups, distributed nodes interact with the physical environment through sensors and actuators, perform local computations, and coordinate each other via a shared communication network. Their decentralized nature exposes CPSs to local exogenous disturbances. For instance, system measurements, which provide information about the system state, can be significantly affected by external phenomena arising from the physical world. This includes human-in-the-loop interactions, which may either be unintentional, as a consequence of human–machine coexistence, or malicious, as in non-invasive false-data injections aimed at tampering with the measurements.

The study of cyber-physical security and attack resilience has attracted significant interest in recent years. Consider a large-scale CPS with a large number of sensor nodes deployed across a wide geographic area. Such a configuration is highly vulnerable to physical and cyber intrusions targeting a subset of sensors. A common attack scenario involves the stealthy introduction of a bias signal that alters the measurements of some sensor nodes without damaging the physical hardware; see, e.g., \cite{sho13}. Because CPSs act directly on physical processes, these disruptions can have serious consequences for both infrastructure safety and human lives; see, e.g., \cite{shi24,mak21}.

Over the last decade, substantial research has focused on quantifying the vulnerability of CPSs, including the maximum number of compromised sensor nodes that a system can tolerate;  see \cite{pas13}. A key challenge in this domain is the observability of systems under sensor data-injection attacks \cite{cho15}. This has led to the formalization of the concept of sparse observability; see \cite{sho16}. Within this framework, the attack signal can be treated either as an unknown exogenous input or as an augmented state vector to be estimated.

The problem of secure state estimation (SSE) can be formulated as the simultaneous reconstruction of both the CPS state and the exogenous/attack signal. In the literature, the sparsity assumption has been widely used to characterize the fundamental limitations of CPSs under sensor attacks; see \cite{pas13,faw14,sho16}. To exploit this sparsity, numerous SSE strategies employ $\ell_1$-norm minimization or regularization; \cite{faw14,paj17}.

The majority of existing observers for SSE rely on batch approaches, which collect a sufficiently large number $\tau$ of time measurements to estimate the $\tau+1$-delayed state. In this framework, \cite{sho16} propose an hard-thresholding observer with an event-triggered projection onto a sparse subspace.  \cite{lu19} investigate a set-cover approach to address the combinatorial complexity of identifying the attacked sensors and develop a switched observer. \cite{fox24} exploit soft-thresholding to design a sparse observer. 

However, {\it{online}} observers, 
which avoid delays and do not accumulate multiple measurements, are highly desirable in practical applications where a timely attack identification  is vital for safety.  

Theoretical guarantees for stability and convergence of online sparse observers remain limited, even in the presence of time-invariant attacks/biases. Although seemingly simple, the resulting system may be unobservable, making standard LTI Luenberger observers unsuitable for joint state and attack reconstruction.

While the observers discussed above are effective in mitigating the effects of sensor data-injection attacks even when $\tau=1$, their behavior in the presence of time-invariant biases has not been deeply investigated.

We remark that online SSE in the presence of sporadic (i.e., sparse in time) perturbations or attacks that occur unpredictably at random intervals has been studied for continuous-time and discrete-time linear time-invariant (LTI) systems by \cite{ale18} and \cite{zam25}, respectively.


In this paper, we study a class of iterative sparse optimization algorithms and adapt them as online observers for CPSs affected by unknown time-invariant sensor injections. Specifically, we investigate approaches based on proximal gradient descent, proximal alternating minimization, Bregman iterations, and the Kaczmarz method. Building on these methodologies, we tailor them to the SSE problem in CPSs and develop novel online estimation strategies.

The paper is organized as follows. In Section \ref{sec:PS}, we formulate the problem and provide the necessary theoretical background. In Section \ref{sec:PA}, we introduce the proposed online sparse observer algorithms. Section \ref{sec:NR} presents  numerical experiments that validate our approach, and Section \ref{sec:C} draws conclusion.

\section{Problem statement and observability}\label{sec:PS}
As in \cite{faw14,sho16,fox24}, we consider CPSs modeled as
\begin{equation}\label{system}
\left\{
 \begin{split}
  &x(k+1)=Ax(k)\\
  &y(k)=Cx(k)+a(k)
 \end{split}\right.
\end{equation}
where $k\in \Z_+=\{0,1,2,\dots,\}$, $A\in\R^{n,n}$, $C\in\R^{q,n}$ with $q<n$; $x(k)\in\R^n$ is the state of the CPS; $y(k)\in\R^q$ is the measurement vector; $a(k)\in\R^q$ represents an unknown exogenous disturbance vector accounting for malicious attacks, operational anomalies, and human-in-the-loop interactions. For simplicity, we refer to $a(k)$ as the {\it{attack vector}}, even though its effects may not originate from intentional adversarial actions. Each sensor $i$ takes a scalar measurement $y_i(k)$; if $a_i(k)\neq 0$, sensor $i$ is under attack.
We assume that $a(k)$ is sparse, i.e., there are few faulty sensors.
For brevity, we omit possible known inputs. Their inclusion would not alter the substance of the analysis; see, e.g., \cite{sho16}.

We consider the following assumption.
\begin{assumption}\label{ass}
The attack vector is time-invariant, i.e., for each $k=0,1,\dots$, $a(k)=\atrue\in\R^q$.
\end{assumption}

Our goal is to address the following online SSE problem.
\begin{problem}\label{p1}
At each time step $k$, estimate $x(k)$ and $\atrue$, given $A$, $C$ and $y(k)$.
\end{problem}

Even under Assumption \ref{ass}, the CPS \eqref{system} may be not observable. In fact, it behaves as a composite system as defined by  \cite{dav75}:  if the constituent subsystems share common eigenvalues, the composite system loses observability due to the indistinguishability of the output modes.
Let us define $z(k)=\begin{pmatrix} x(k)^\top,& a(k)^\top\end{pmatrix}^\top$, $\Aa = \begin{pmatrix}
    A & 0\\
    0 & I_q
       \end{pmatrix}
$ and $\Ca = \begin{pmatrix}
    C & I_q
       \end{pmatrix}$, where $I_q$ is the identity matrix of dimension $q$. Then, if $a(k)=\atrue$, CPS \eqref{system} can be written in the compact form
\begin{equation}\label{system2}
\left\{
 \begin{split}
  &z(k+1)=\Aa z(k)\\
  &y(k)=\Ca z(k)
 \end{split}\right.
\end{equation}       
The following result specializes Theorem 1 by \cite{dav75} in our framework. For completeness, we include a dedicated proof.        
\begin{proposition}
Let the attack-free CPS be observable, i.e., the pair $(A,C)$ is observable, and let $q<n$.
If $A$ has the eigenvalue $1$, the CPS \eqref{system2} with constant attacks is not observable, i.e., the pair $(\Aa,\Ca)$ is not observable.
\end{proposition}
\begin{pf}
We define the observability matrix of the augmented system described by $(\Aa,\Ca)$ as
\begin{equation}\label{eq:ob}
 \O = \begin{pmatrix}
 C&I_q\\
 CA&I_q\\
 \vdots&\vdots\\
 CA^{\tau-1}&I_q
\end{pmatrix}
\end{equation}
where $\tau$ is the number of acquired measurements. By suitable row reduction of $\O$, we obtain
\begin{equation}\label{eq:ob2}
 \begin{pmatrix}
  C&I_q\\
 C(A-I_n)&0\\
 \vdots&\vdots\\
 CA^{\tau-2}(A-I_n)&0
\end{pmatrix}
\end{equation}
If $A$ has the eigenvalue $1$, then  $A-I_n$ has the eigenvalue $0$. Thus, the matrix in \eqref{eq:ob2} and $\O$ are not full rank, and the system is not observable.
\end{pf}

The case where $A$ possesses an eigenvalue equal to $1$ encompasses all LTI dynamical systems characterized by a constant natural mode. For instance, stochastic matrices - widely employed to model, e.g., consensus dynamics in multi-agent systems, Markov chains, and web search algorithms - exhibit a Perron-Frobenius eigenvalue equal to $1$. Another example of a system matrix featuring this property is the IEEE 14-bus power network, a widely adopted benchmark in CPS literature; see, e.g., \cite{pas13}. Nevertheless, this unobservable case has not been examined in the literature. For example, \cite{sun26} exclude it from their analysis.

\subsection{Basis Pursuit and Lasso formulation}
If $a(k)=\atrue$ for each $k$, a batch approach to SSE consists in solving $y=\O z$ to recover $z(0)=\begin{pmatrix}x(0)^\top,&\atrue^\top\end{pmatrix}^\top.$, where $y=\begin{pmatrix} y(0)^\top,&\dots,& y(k-1)^\top\end{pmatrix}^\top$.
Nevertheless, if CPS \eqref{system2} is not observable this equation has infinitely many solutions and the estimation fails. To address this issue, one can exploit the attack sparsity by introducing an $\ell_1$  minimization or regularization. The $\ell_1$ norm is the best convex approximation of the sparsity level of a vector and provides an effective tool to sparsify the solution of an underdetermined linear system. This framework has been widely studied in compressed sensing, where vectors admitting sparse representations are reconstructed from sub-Nyquist measurements via convex optimization problems, such as Basis Pursuit and Lasso; see \cite{fou13}. We notice that subsampling in static systems structurally mirrors unobservability in dynamic systems: both yield non-injective observation mappings that hide critical state information in the null space.

Let $\reg(z)=\sum\limits_{i=1}^{n+q}\Lambda_i|z_i|$ be the weighted $\ell_1$ norm. In our setting, the $\ell_1$ norm must be applied only on $a$, thus we naturally set $\Lambda_i = \lambda>0$ for $i=1,\dots,n$, and $\Lambda_i=0$ for $i>n$. 

The Basis Pursuit formulation is
\begin{equation}\label{bp}
 \zstar=\argmin{z\in\R^{n+q}} ~~\reg(z) ~\text{ s.t.} ~y = \O z.\end{equation}
The Lasso formulation is
\begin{equation}\label{lasso}
\zstar=\argmin{z\in\R^{n+q}}~~ \frac{1}{2}\left\|y - \O z\right\|_2^2+\reg(z).
\end{equation}
 The Basis Pursuit is well-posed under the null-space property, see \cite{fou13}, while Lasso guarantees the correct identification of the attacks under the irrepresentable condition; see  \cite{fox24}.
The least-squares formulation of Lasso  may envisage the presence of measurement noise. In this work, we focus on the noise-free case, then both formulations are suitable.

Our goal is to analyze and develop methods for the online solution of problems \eqref{bp} and \eqref{lasso}.

\section{Online sparse observers}\label{sec:PA}
We design online sparse observers by starting from recursive methods to solve problems \eqref{lasso} and \eqref{bp}. An observer is said to be "sparse" if it exploits prior information on sparsity; this generally yields nonlinear observers. Moreover, we call "online" an observer that uses only the current measurement $y(k)$ to estimate the current state.

The setting is as follows. Each time measurement provides a subset of equations of $y=\O z$, which are immediately processed and then eliminated from the storage. We apply recursive methods that, each $k$, provide new estimates $\xhat(k)$ of $x(0)$ and $\ahat(k)$ of  $\atrue$; then we estimate the current state $x(k)$ as $A^k \xhat(k)$.

\subsection{Online proximal gradient descent}
Lasso can be solved by proximal gradient descent (PGD), also known as iterative soft-thresholding algorithm; see \cite{dau04}. When  measurements are streaming, an online version of PGD, hereafter denoted as O-PGD, can be applied. \cite{duc10} analyze O-PGD in terms of static regret, by assuming that the dynamics of the system is unknown and revealed at each $k$.


Let us define
\begin{equation}
\fun_k(z) =  \frac{1}{2}\|\Ca \Aa^k z-y(k)\|_2^2 = \frac{1}{2}\|C A^k x+a-y(k)\|_2^2.
\end{equation}
Then, the Lasso problem  \eqref{lasso} can be written as
\begin{equation}\label{lasso2}
\zstar=\argmin{z\in\R^{n+q}} \frac{1}{\tau}\sum_{k=0}^{\tau-1} \fun_k(z) +\reg(z).
\end{equation}
O-PGD iteration consists in minimizing a first-order approximation of $\fun_k$ around the current estimate $\zhat(k)$, plus $r(z)$ and a proximal regularization term $\frac{\eta}{2} \|z-\zhat(k)\|_2^2$, $\eta>0$.
In formulas,
\begin{equation}\label{opgd}
\begin{split}
&\zhat(k+1)=\\
&= \argmin{z\in\R^{n+q}} \langle \nabla \fun_k(\zhat(k)),z \rangle+\reg(z)+ \frac{\eta}{2} \|z-\zhat(k)\|_2^2\\
 &=  \soft_{\frac{\lambda}{\eta}}\left( \zhat(k) - \frac{1}{\eta}(\Ca \Aa^k)^\top \left(\Ca \Aa^k \zhat(k)-y(k)\right) \right)\\
\end{split}
\end{equation}
where  $\soft$ is the componentwise soft-thresholding operator, defined by $\soft_{\nu}(w)=w-\nu\sign(w) $ if $|w|\geq \nu$, and $0$ otherwise, for $w\in\R$. Throughout the paper, we use the notation $\soft_{\nu}(z)$ by meaning that soft-thresolding applies only to the last $q$ components of $z$, as the $\ell_1$ regularization is only on the attack part of the augmented state.

We notice that by enforcing closeness to $\zhat(k)$, proximal regularization transfers information from past measurements to the current update and prevents the algorithm from discarding knowledge accumulated from earlier observations.
\cite{duc10} prove that O-PGD enjoys a sublinear static regret, meaning that  $\fun_k(\zhat(k)) +\reg(\zhat(k))$ converges on average to $\frac{1}{\tau}\sum\limits_{k=0}^{\tau-1} \fun_k(\zstar) +\reg(\zstar).$


\subsection{Online proximal alternating minimization}
While O-PGD exploits the linearization of $\fun_k$ to deal with the presence of the $\ell_1$ regularization, we should notice that in our Lasso problem the variable $a$, interested by the $\ell_1$ regularization, is decoupled in $\fun_k$. Therefore, we can proceed without linearization, by alternating minimization over $x$ and $a$. Both minimizers can be uniquely computed with a straightforward closed form, provided that the proximal  regularization term $\frac{\eta}{2} \|z-\zhat(k)\|_2^2$ is added. The convergence of this proximal alternating minimization (PAM) method is proven by \cite{ban18} in the Jacobi form and by \cite{att08} in the Gauss-Seidel form for batch problems. 

We propose to extend Jacobi PAM for online SSE by using only the current $\fun_k$; we denote this approach as online PAM (O-PAM). Let us define  $$M_k=\left((CA^k)^\top CA^k+\eta I_n\right)^{-1}\in\R^{n,n}.$$

Then, O-PAM iteration is as follows.
\begin{equation}\label{opam}
 \begin{split}
\xhat(k+1)&=\argmin{x\in\R^n}\fun_k(x,\ahat(k))+\frac{\eta}{2}\|x-\xhat(k)\|_2^2\\
 &= M_k\left(\eta \xhat(k) +(CA^k)^\top(y(k)-\ahat(k))\right)\\
 &= \xhat(k)+ M_k(CA^k)^\top \left(y(k)-\ahat(k)-CA^k\xhat(k)\right),\\
 &\\
\ahat(k+1)&=\argmin{a\in\R^q}~\fun_k(\xhat(k),a)+\lambda\|a\|_1+\frac{\eta}{2}\|a-\ahat(k)\|_2^2\\
 &=\soft_{\frac{\lambda}{1+\eta}}\left( \frac{y(k)-CA^k \xhat(k)+ \eta \ahat(k)}{1+\eta} \right)\\
 &= \soft_{\frac{\lambda}{1+\eta}}\left(\ahat(k)+ \frac{y(k)-\ahat(k)-CA^k\xhat(k)}{1+\eta} \right).
 \end{split}
\end{equation}
In compact form, O-PAM is
\begin{equation}\label{OPAMcompact}
\zhat(k+1)=\soft_{\frac{\lambda}{1+\eta}}\left(\zhat(k)- P_k \left(\Ca \Aa^k \zhat(k)-y(k)\right)\right)
\end{equation}
where $$P_k=\begin{pmatrix}
M_k (CA^k)^\top\\ 
\frac{1}{1+\eta}I_q
\end{pmatrix}.$$
Interestingly, O-PAM in \eqref{OPAMcompact} and O-PGD in \eqref{opgd} share the same structure. The difference is that, in O-PAM, $P_k$ replaces $(\Ca \Aa^k)^\top$ used in O-PGD. More specifically, O-PAM employs a $\eta$-regularized pseudo-inverse $M_k(CA^k)^\top$ of $CA^k$ rather than its transpose $(CA^k)^\top$. 

Since O-PAM does not rely on a first-order approximation, it is expected to converge in fewer iterations than O-PGD. However, each iteration of O-PAM requires inverting an $n \times n$ matrix, which may limit its scalability to large-dimensional CPSs.


\subsection{Block linearized Bregman and sparse Kaczmarz}
The linearized Bregman method  addresses problem \eqref{bp}. While the structure of the algorithm is very similar to PGD, Linearized Bregman converges to the minimum of an elastic-net Basis pursuit, where the cost function is the sum of $\ell_1$ norm and squared $\ell_2$ norm; see \cite{yin08,osh10}.

\cite{lor14, lor14b} extend Bregman methods to  online sparse optimization problems with streaming data and establish convergence guarantees for a family of block linearized Bregman methods, including sparse Kaczmarz algorithms. These results are also applicable to our framework.

The block linearized Bregman method corresponds to O-PGD in \eqref{opgd} with the $\ell_1$ term replaced by the $\ell_1$ Bregman divergence $\|a\|_1-\|\ahat(k)\|_1-\langle a-\ahat(k),p(k) \rangle$, $p(k)$ being a subgradient of $\|\ahat(k)\|_1$.
By straightforward algebraic manipulations - see, e.g., Section 5.3 of \cite{yin08} - the algorithm can be expressed as
\begin{equation}\label{linbreg}
\zhat(k+1)=   \zhat(k) - \frac{1}{\eta}(\Ca \Aa^k)^\top \left(\Ca \Aa^k ~\soft_{\frac{\lambda}{\eta}}\left(\zhat(k)\right)-y(k)\right).
\end{equation}
In a nutshell, the position of the soft-trhresholding operator changes with respect to O-PGD.

Based on this consideration, we propose the following Bregman variant of O-PAM (hereafter denoted as PAM Bregman):
\begin{equation}\label{bregpam}
\zhat(k+1)= 
    \zhat(k) - P_k \left(\Ca \Aa^k ~\soft_{\frac{\lambda}{1+\eta}}\left(\zhat(k)\right)-y(k)\right).
\end{equation}

In \eqref{linbreg}, $\soft_{\frac{\lambda}{\eta}}\left(\zhat(k)\right)$ converges to the minimum of the elastic-net version of \eqref{bp}; see \cite{lor14}.

By the same arguments used in Theorem 5.3 by \cite{osh10}, we can prove that, if it converges, PAM Bregman achieves a solution of $y=\O z$. On the other hand, proving its convergence is challenging because $P_k \Ca \Aa^k$ is not symmetric as $(\Ca \Aa^k)^\top(\Ca \Aa^k)$. As a result, the arguments used in Theorem 3.2 of \cite{osh10} cannot be applied directly.

The sparse Kaczmarz algorithm is a variant of \eqref{linbreg} that updates the estimate by using a single scalar measurement at each iteration; see \cite{lor14,lor14b} for details.

\subsection{Aggregated O-PGD and O-PAM}
In the framework of sequential processing of large datasets in machine learning, \cite{van18} propose an aggregated O-PGD method (AO-PGD) which achievs linear convergence. The main idea of this approach is to aggregate past information by considering the gradient of $\sum_{h=0}^{k}\fun_h$ instead of a single $\fun_k$. This requires no extra memory, but the increment of $\sum_{h=0}^k (\Ca \Aa^h)^\top (\Ca \Aa^h)$ and $\sum_{h=0}^k (\Ca \Aa^h)^\top y(h)$ at each iteration. The same idea can be extended to O-PAM, by considering $P_h$ instead of $(\Ca \Aa^h)^\top$; we call this variant AO-PAM.



\section{Numerical results}\label{sec:NR}
In this section, we assess the performance of the proposed online observers through Monte Carlo
simulations. We compare O-PGD, O-PAM, their aggreated versions AO-PGD and AO-PAM, 
the block linearized Bregman method (Lin.~Bregman), its PAM variant (PAM Bregman), and the sparse Kaczmarz observer. 

At each iteration $k$, every algorithm updates the estimate $\zhat(k)=(\xhat(k)^\top,\ahat(k)\top)^\top$ of $x(0)$ and $\atrue$, and reconstructs the current state as
$A^{k}\xhat(k)$.
\subsection{Experimental setup}
We consider three families of CPSs, all sharing the feature of
Proposition~1.

(i)~~\emph{Multi-agent dynamics.} We consider a multi-agent network with $n=50$ nodes. The topology is illustrated in Fig. \ref{fig:g}: the network consists of two ring clusters and each cluster has a leader. $A$ is a row-stochastic matrix
associated with this topology, and, as a consequence, has a double eigenvalue equal to 1.
The system dynamics, shown in Fig. \ref{fig:ev}, converge to consensus on the leader's state within each cluster. We assume that $q=20$ sensors measure the state through the matrix $C$ that has i.i.d. Gaussian entries; moreover, $h=4$ sensors are subjected to attack. For this family the triple $(A,C,x(0))$
is kept fixed, while the location of the attacked sensors is randomized across the runs.

(ii)~~\emph{Marginally stable dynamics.} $A$ is randomly generated and normalized to have a unique eigenvalue equal to 1 and all the other eigenvalues strictly inside the unit disk.

(iii) \emph{Mildly unstable dynamics.} $A$ has an 
unstable eigenvalue ($\irr(A)\approx 1.0005$), two marginally stable
modes ($\pm 1$), and the remaining eigenvalues strictly inside the
unit disk.

For (ii)--(iii) a new realization of $A$ is drawn at each
run. In all the experiments $C\in\R^{q\times n}$ has i.i.d.\ Gaussian
entries (with $n=40$, $q=20$ for (ii)--(iii); for family~(i) the
dimensions depend on the network model), the attack support is
drawn uniformly at random with cardinality $h=5$, and the non-zero
attack entries and the entries of $x(0)$ have random signs and
magnitudes in $[1,2]$ and $[2,3]$, respectively. We report noise-free
experiments. We set to $\lambda=0.6$
for the proximal-gradient and alternating-minimization methods; for the
Bregman and Kaczmarz observers, we consider $20\lambda$. All curves are averaged over $50$ independent runs.

\begin{figure}[t]
\centering
\includegraphics[width=0.94\columnwidth]{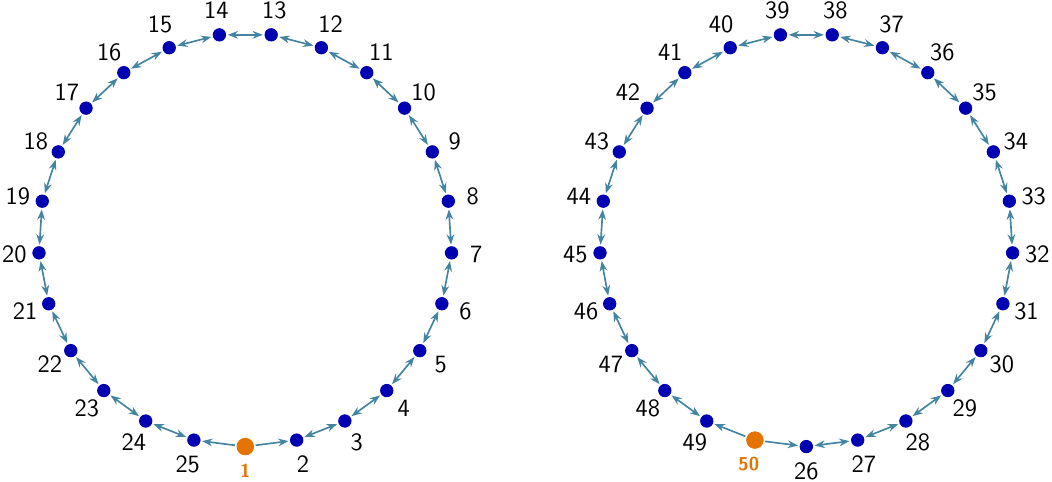}
\caption{Example (i): topology of the network. We consider two disconncted subnetworks with ring topology and a leader. The corresponding matrix $A\in\R^{n,n}$, $n=50$, is  row-stochastic with a double eigenvalue equal 1.}
\label{fig:g}
\end{figure}

\begin{figure}[t]
\centering
\includegraphics[width=0.94\columnwidth]{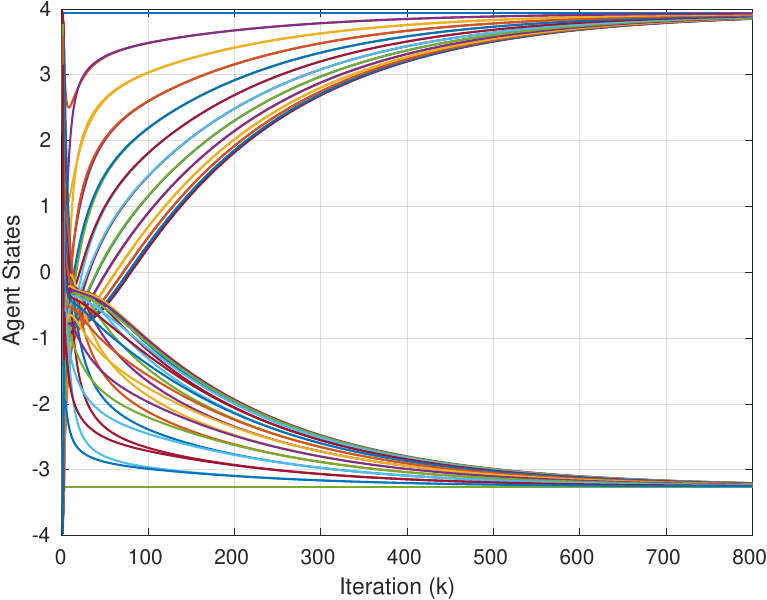}
\caption{Example (i): dynamics of the multi-agent CPS.}
\label{fig:ev}
\end{figure}

\subsection{Performance metrics} 
At each $k$ we monitor two quantities. The first is the mean squared state
estimation error $\tfrac{1}{n}\|x(k)-\xhat(k)\|_2^2$. The second is the quality of the support recovery of the attack: a sensor is declared
attacked when $|\ahat_i(k)|>\theta$ for a fixed threshold $\theta$. We define the support attack error as the total number of
misclassified sensors, i.e., the sum of false positives (sensors erroneously flagged as attacked) and false negatives (missed attacks).

\subsection{Results and discussion}
In examples (i) and (ii),
all the seven observers successfully identify the attacked sensors:
the support error drops to zero within a few tens of iterations
(Figs.~\ref{fig:caso_1_support},~\ref{fig:caso_2_support}),
confirming that $\ell_1$ regularization can compensate for the lack of
observability established by Proposition~1. On both examples,
the algorithms differ markedly in state estimation accuracy
(Figs.~\ref{fig:caso_1_state},~\ref{fig:caso_2_state}): as guaranteed by the theory, the
Bregman/Kaczmarz observers (Lin.~Bregman, PAM~Bregman, Sparse
Kaczmarz) are unbiased, and achieve null estimation error once that attacks are correctly identified.



The mildly unstable family~(iii) provides the most challenging
scenario, as the unstable mode amplifies any residual estimation error
through $A^k$. On this family the three Bregman/Kaczmarz observers
stand out: Lin.~Bregman, PAM~Bregman, and Sparse Kaczmarz recover
the attack support exactly and achieve the lowest state estimation
error by a wide margin
(Figs.~\ref{fig:caso_3_state},~\ref{fig:caso_3_support}). The
remaining four methods still converge, but retain a small residual
support error and a correspondingly higher state estimation floor,
reflecting the progressive ill-conditioning of the per-step Gram
matrix $(\Ca\Aa^k)^\top(\Ca\Aa^k)$ to which their updates are more
sensitive.
Overall, the Bregman/Kaczmarz observers consistently provide the best
trade-off between support-recovery accuracy and state estimation
accuracy across all the three families, and are the only methods that
achieve exact support recovery on the unstable family~(iii). The
proximal-gradient and alternating-minimization methods perform well on stable dynamics, but exhibit a mild degradation on unstable dynamics.

Finally, Figs. \ref{fig:FP1}, \ref{fig:FN1}, \ref{fig:FP2}, \ref{fig:FN2}, \ref{fig:FP3}, and \ref{fig:FN3} report the numbers of false positives and false negatives for all the considered examples. These results highlight that Bregman-based approaches reduce the number of false positives during the transient phase if compared to the other methods. This behavior stems from the unbiased nature of the Bregman formulation, which allows the weight of the $\ell_1$ regularization term to remain relatively large without compromising the final estimate. Consequently, Bregman-based algorithms are more conservative when declaring that a sensor is under attack, and such declarations are typically definitive. This characteristic is particularly desirable in applications where the countermeasures triggered by attack detection are irreversible.

\begin{figure}[t]
\centering
\includegraphics[width=0.94\columnwidth]{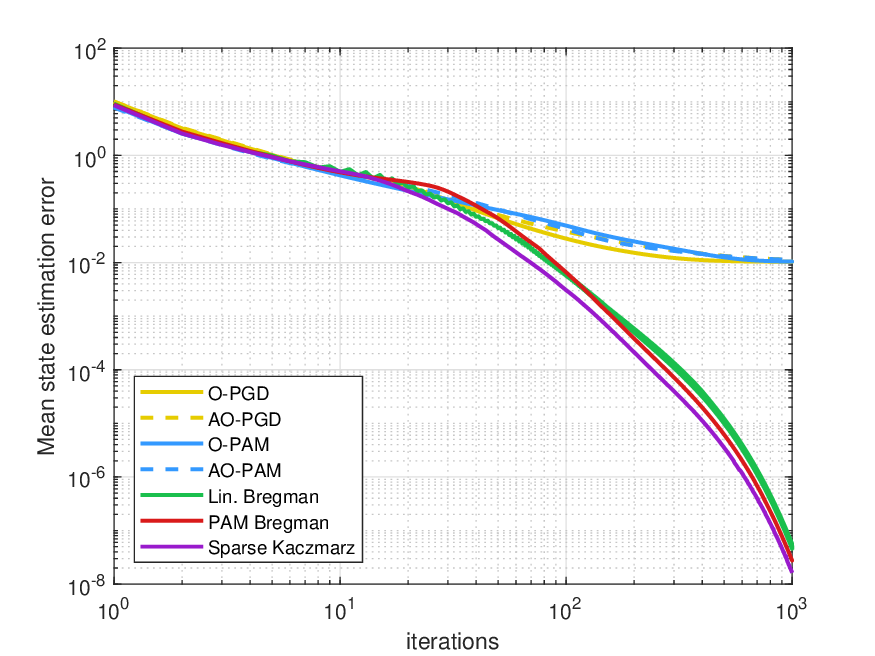}
\caption{Example (i): state estimation error.}
\label{fig:caso_1_state}
\end{figure}
 
\begin{figure}[t]
\centering
\includegraphics[width=0.94\columnwidth]{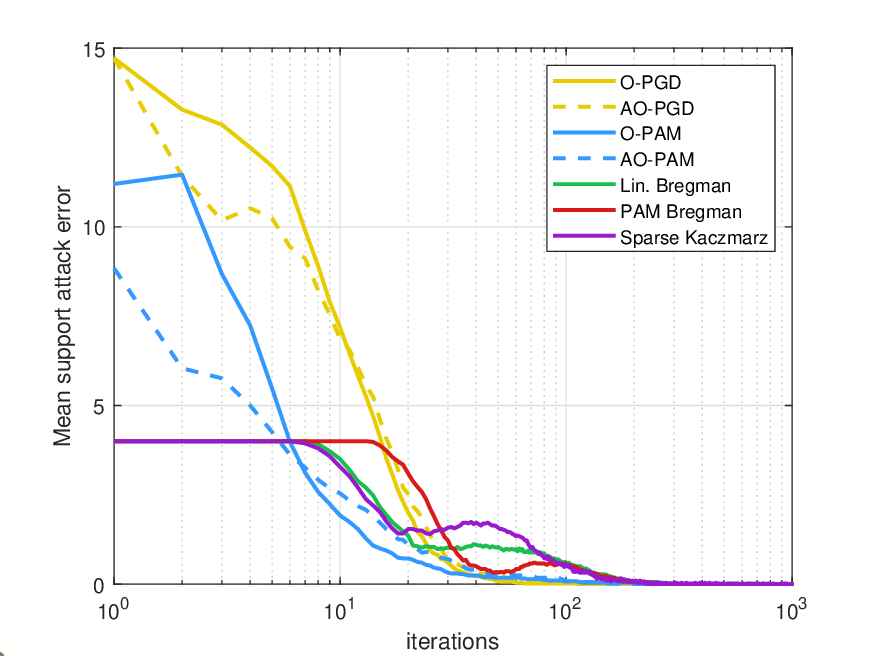}
\caption{Example (i): support attack error.}
\label{fig:caso_1_support}
\end{figure}

\begin{figure}[t]
\centering
\includegraphics[width=0.94\columnwidth]{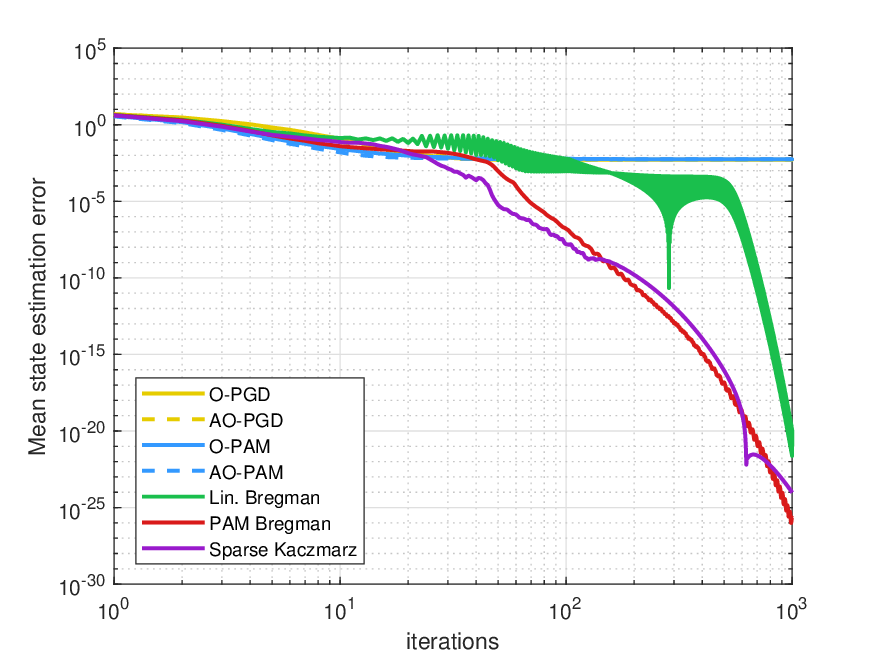}
\caption{Example (ii): state estimation error.}
\label{fig:caso_2_state}
\end{figure}
 
\begin{figure}[t]
\centering
\includegraphics[width=0.94\columnwidth]{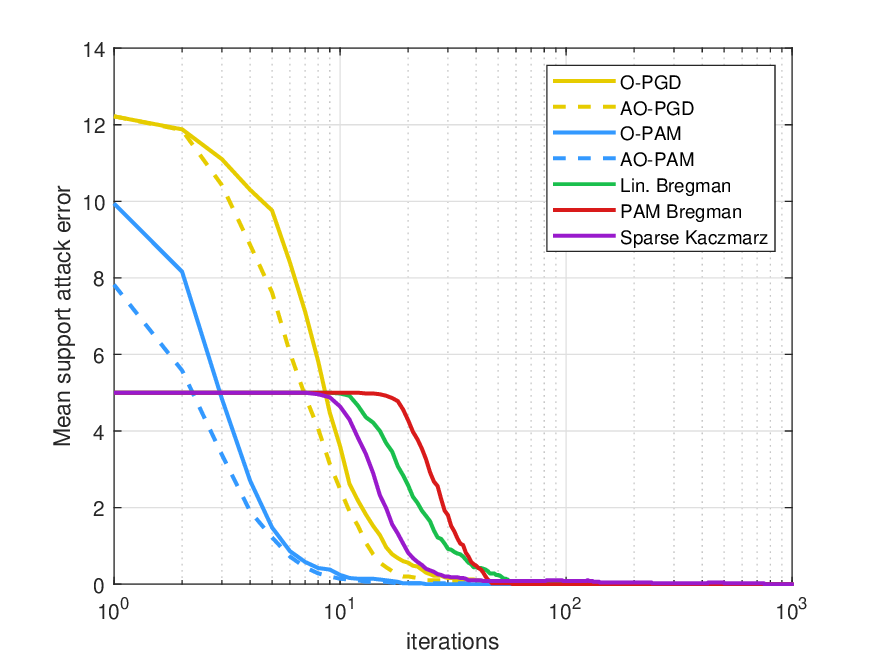}
\caption{Example (ii): support attack error.}
\label{fig:caso_2_support}
\end{figure}

\begin{figure}[t]
\centering
\includegraphics[width=0.94\columnwidth]{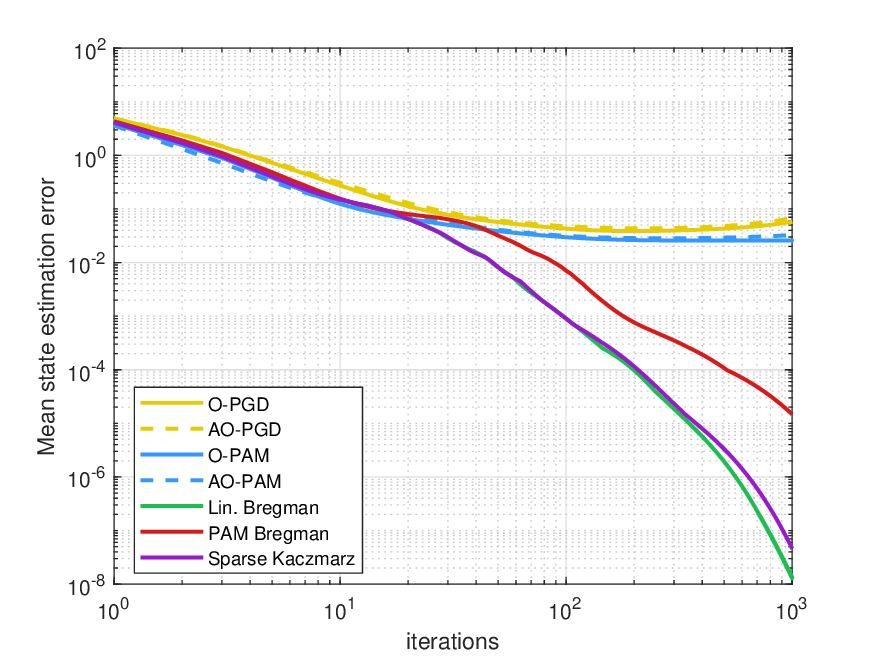}
\caption{Example (iii): state estimation error.}
\label{fig:caso_3_state}
\end{figure}
 
\begin{figure}[t]
\centering
\includegraphics[width=0.94\columnwidth]{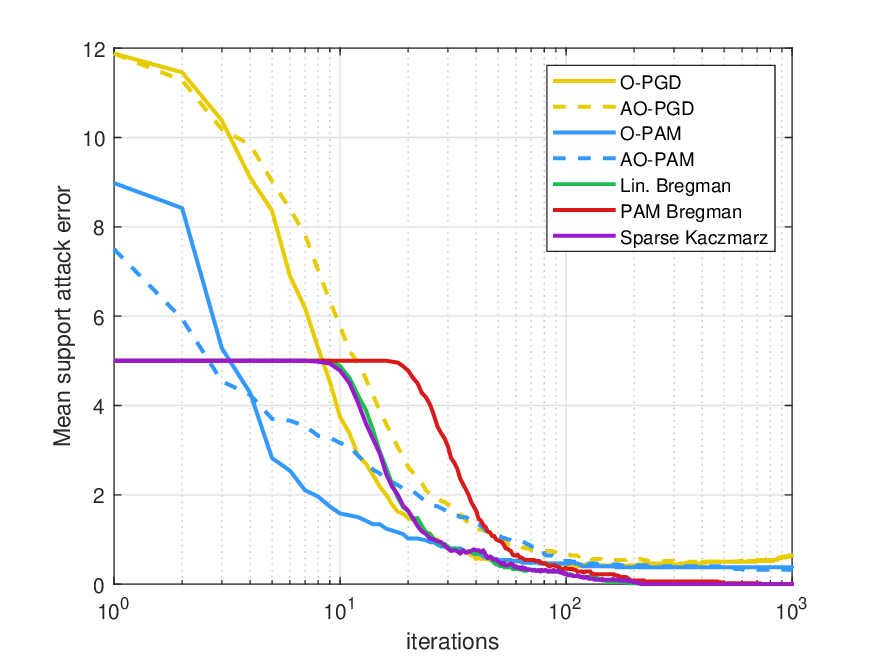}
\caption{Example (iii): support attack error.}
\label{fig:caso_3_support}
\end{figure}

\begin{figure}[t]
\centering
\includegraphics[width=0.94\columnwidth]{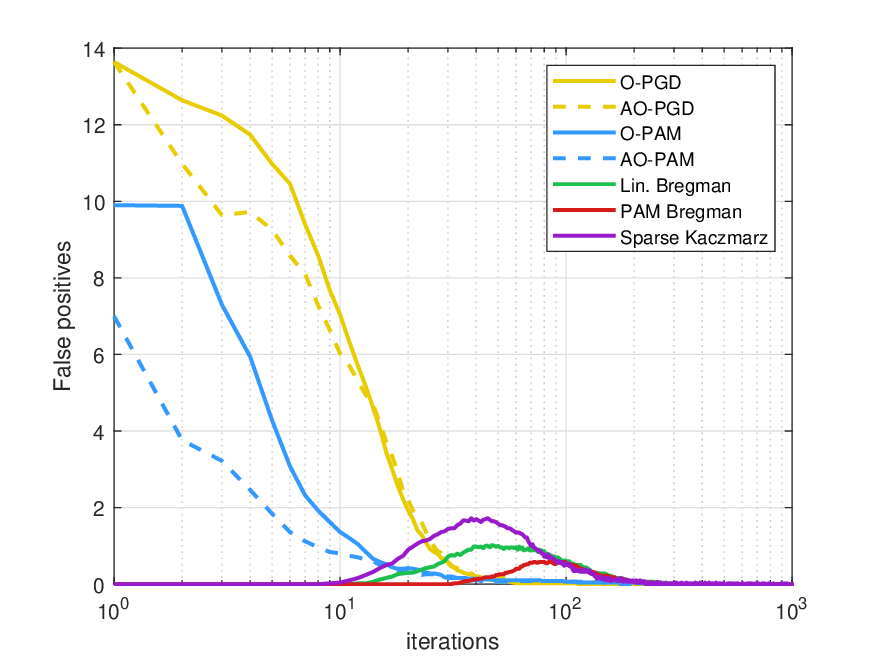}
\caption{Example (i): false positives.}
\label{fig:FP1}
\end{figure}
 
\begin{figure}[t]
\centering
\includegraphics[width=0.94\columnwidth]{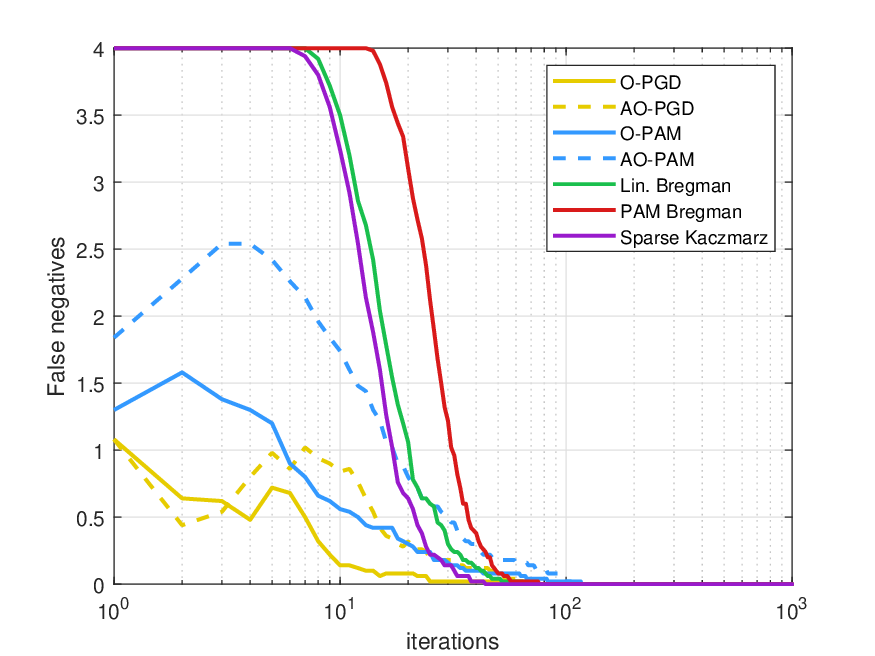}
\caption{Example (i): false negatives.}
\label{fig:FN1}
\end{figure}

\begin{figure}[t]
\centering
\includegraphics[width=0.94\columnwidth]{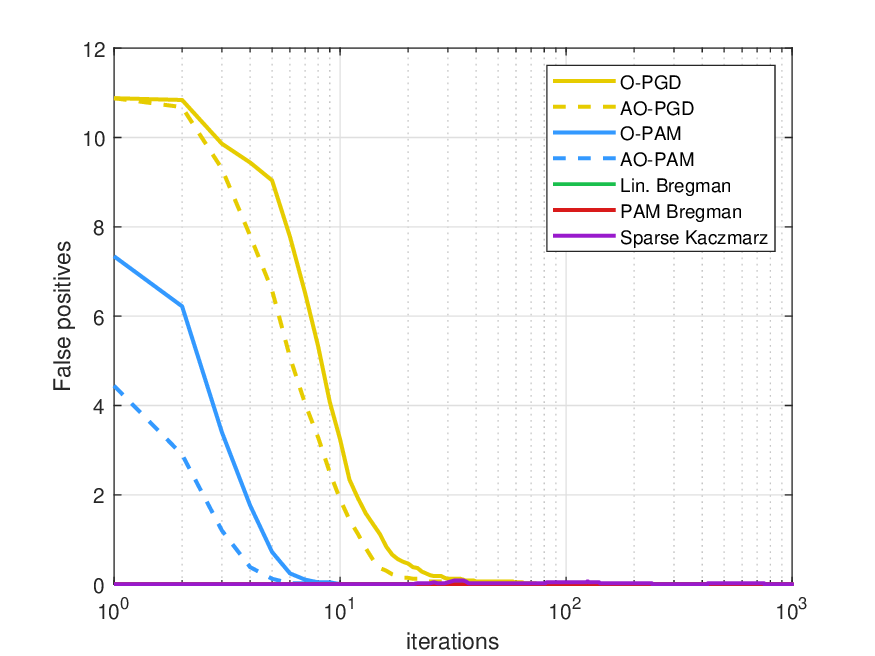}
\caption{Example (ii): false positives.}
\label{fig:FP2}
\end{figure}
 
\begin{figure}[t]
\centering
\includegraphics[width=0.94\columnwidth]{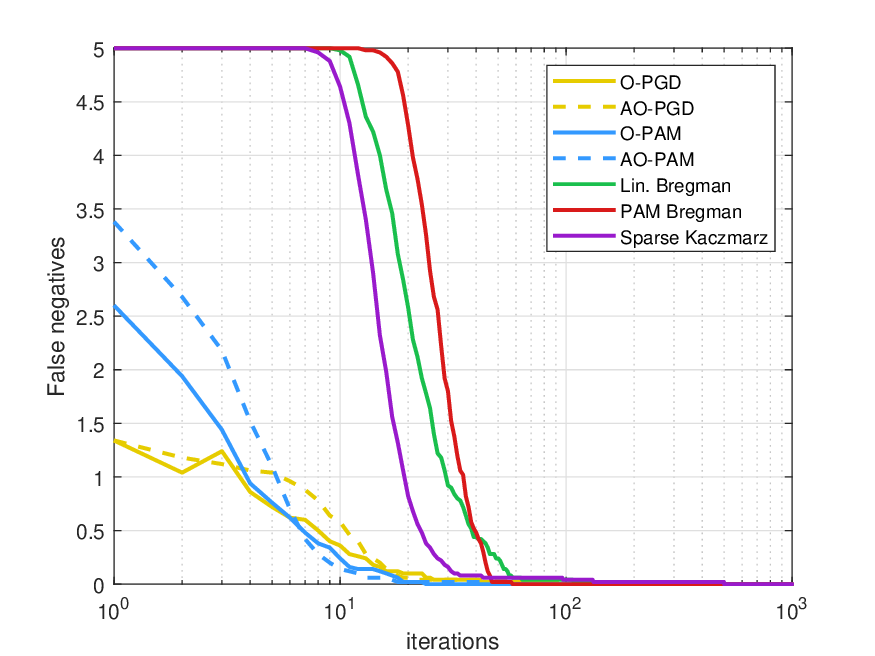}
\caption{Example (ii): false negatives.}
\label{fig:FN2}
\end{figure}

\begin{figure}[t]
\centering
\includegraphics[width=0.94\columnwidth]{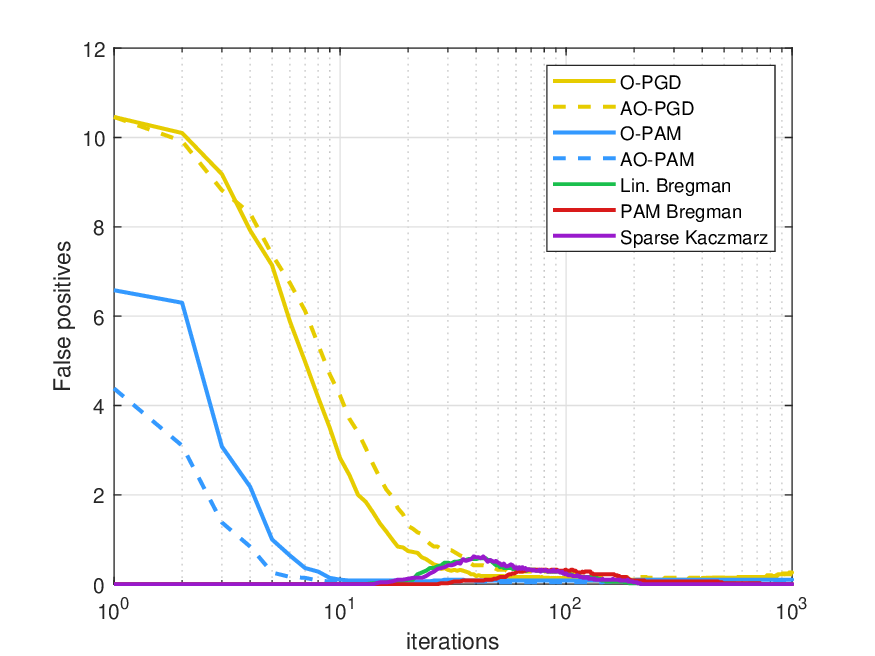}
\caption{Example (iii): false positives.}
\label{fig:FP3}
\end{figure}
 
\begin{figure}[t]
\centering
\includegraphics[width=0.94\columnwidth]{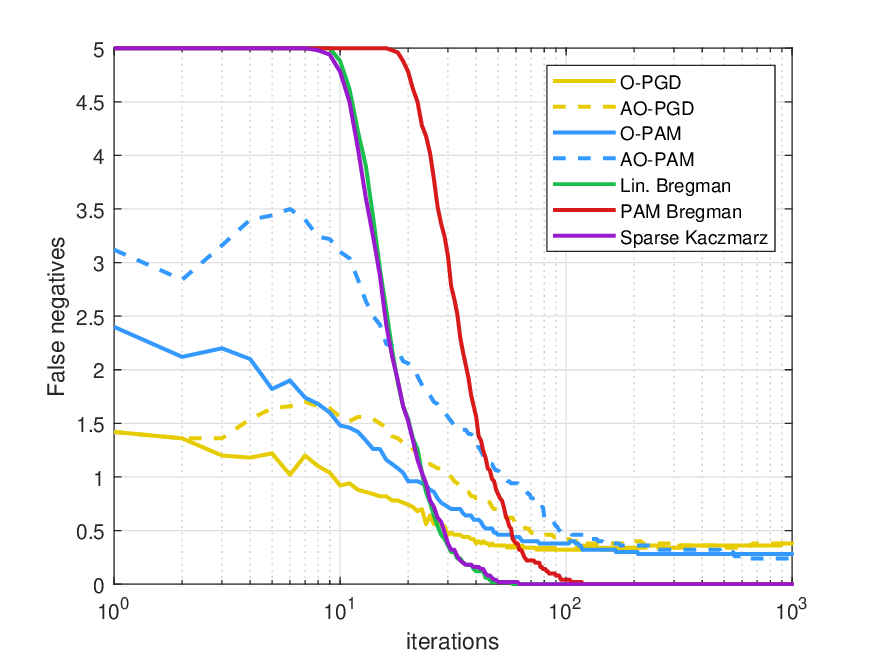}
\caption{Example (iii): false negatives.}
\label{fig:FN3}
\end{figure}

\section{Conclusions}\label{sec:C}
In this paper, we analyze a class of sparse observers for online secure state estimation of cyber-physical systems. We consider systems with time-invariant, sparse sensor biases, which can represent, e.g., undesired human intervention or malicious attacks. 

Among the algorithms considered, the observers based on Bregman iterations exhibit superior  estimation accuracy across the numerical experiments. 

The convergence analysis of the proposed Bregman-based proximal alternating minimization algorithm remains an open problem and is the subject of ongoing research.
\bibliography{cps}
\end{document}